\documentclass[submission,copyright,creativecommons]{eptcs}
\providecommand{\event}{TARK XVII} 
\usepackage{underscore}           

\usepackage{graphicx}
\usepackage{amsthm}

\usepackage{amsmath}               
\usepackage{amsfonts}
\usepackage{dsfont}
\usepackage{amssymb} 
\usepackage{bm}
\usepackage[inline]{enumitem}
\usepackage{color, soul}	

\newcommand{\cl}[1]{\mathcal{#1}} 

\newtheorem{thm}{Theorem}

\newtheorem{prop}[thm]{Proposition}

\newtheorem{obs}[thm]{Observation}

\newtheorem{qst}[thm]{Question}
\newtheorem{exm}[thm]{Example}

\newcommand\blfootnote[1]{%
	\begingroup
	\renewcommand\thefootnote{}\footnote{#1}%
	\addtocounter{footnote}{-1}%
	\endgroup
}

\title{Sequential Voting with Confirmation Network\blfootnote{This project has received funding from the European Research Council (ERC) under the European Union's Horizon 2020 research and innovation programme (grant agreement $ n^o $ 740435).}}

\author{\qquad Yakov Babichenko \qquad\qquad\qquad Oren Dean\qquad\qquad\qquad Moshe Tennenholtz
	\email{\quad\quad yakovbab@tx.technion.ac.il \quad orendean@campus.technion.ac.il \quad moshet@ie.technion.ac.il}
	\institute{Faculty of Industrial Engineering and Management\\
		Technion---Israel Institute of Technology\\
		Haifa, ISRAEL}
}

\def\titlerunning{Sequential Voting with Confirmation Network}
\def\authorrunning{Y. Babichenko, O. Dean \& M. Tennenholtz}
\begin{document}
\maketitle

\begin{abstract}
	We discuss voting scenarios in which the set of voters (agents) and the set of alternatives are the same; that is, voters select a single representative from among themselves. Such a scenario happens, for instance, when a committee selects a chairperson, or when peer researchers select a prize winner. Our model assumes that each voter either renders worthy (confirms) or unworthy any other agent. We further assume that the prime goal of each agent is to be selected himself. Only if that is not feasible, will he try to get one of those that he confirms selected. In this paper, we investigate the open-sequential voting system in the above model. We consider both plurality (where each voter has one vote) and approval (where a voter may vote for any subset). Our results show that it is possible to find scenarios in which the selected agent is much less popular than the optimal (most popular) agent. We prove, however, that in the case of approval voting, the ratio between their popularity is always bounded from above by 2. In the case of plurality voting, we show that there are cases in which some of the equilibria give an unbounded ratio, but there always exists at least one equilibrium with ratio 2 at most.	
\end{abstract}

\section{Introduction}
Consider a committee voting to select a chairperson. Each committee member would like the honour of serving as chairperson himself. As a second best option he prefers one of several other members to win the position.\footnote{An equivalent situation arises when the committee is about to select a course of action (e.g., tenure a researcher) and each committee member is strongly associated with one of the alternatives.} The committee members' preferences profile can be represented by a \emph{confirmation network}, in which a direct edge indicates that the source of the edge confirms the target. In the confirmation network of Figure~\ref{fig: cabinet} member \#5 is the \emph{most popular} member --- he is supported by three other members, while everyone else is supported by at most one.\\
\begin{figure}[h]
	\centering
	\includegraphics[width={0.3\textwidth}]{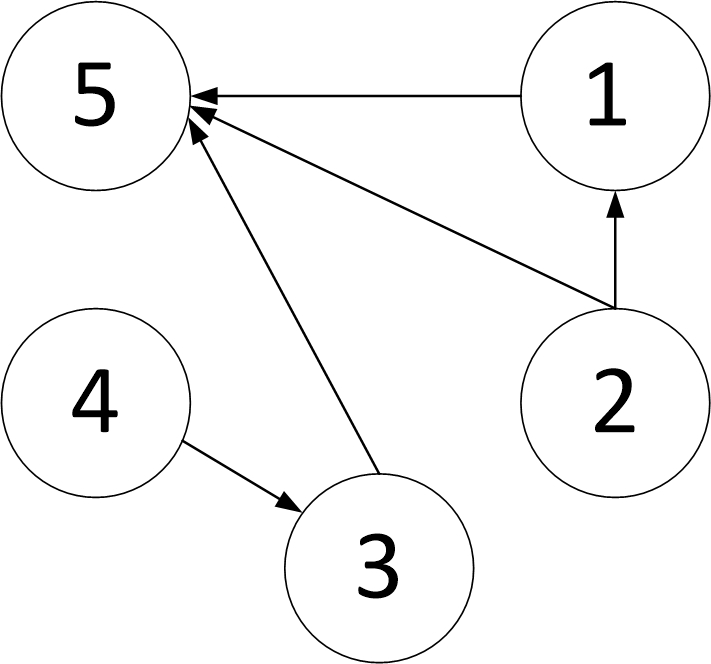}
	\caption{Five committee members and their confirmations.}
	\label{fig: cabinet}	
\end{figure}
As is the case in many committees, we assume the members are well-known to each other, and hence the confirmation network is known to everyone. The ballot is open and conducted sequentially. Since the committee members are always sited in the same places around the table, the voting order is prefixed and known. We consider two voting
methods: plurality, in which
each voter selects only one other member, and approval,
in which members vote for any subset of the other
members. In either voting method, a member is not allowed
to vote for himself, but he is allowed to abstain. The member with the most votes
is elected. Ties are broken by a predetermined and publicly known preference order.

Game-theoretically, we have a multi-stage game, describable as an \emph{extensive-form game} --- a tree with all possible voting-sequences, and an outcome at every leaf. The standard solution for this kind of game is a \emph{subgame perfect equilibrium} (SPE). To find an SPE, we start with the last voter. For every voting history, we assume this voter will choose a ballot which gives him a best outcome (notice that there may be more than one `best outcome'). Moving to the next-to-last voter, we know, for every voting history, how the last voter will respond to any of his ballots. Thus, we can find all his best votes and fix one of them for any sequence of voting history. We can continue this backward process until we select a best vote for all voters. Since voters are indifferent between their best votes, in general there can be many SPEs. 

We exemplify the model and its complexity in two scenarios. Example~\ref{exm: first} shows a case in which the most popular member is not elected in the unique SPE of plurality voting. The same scenario with approval voting leads to two different SPEs --- in one of them the most popular member is the winner. In the network of Example~\ref{exm: second} each member confirms at most one other member. Later (Proposition~\ref{prp: 1-regular}) we will see that under this condition the outcome is always `almost-optimal'. Nevertheless, Example~\ref{exm: second} shows that in the case of approval the outcome is not trivial: one of the members manages to get a better result by voting for someone he does not confirm.

\begin{exm}\label{exm: first}
	\emph{Assume that in the network of Figure~\ref{fig: cabinet} the voting order is lexicographic, and so is the tie-breaking order.\footnote{I.e., in case of a draw, the voter with lowest index wins.} We will show that in this case, we have a unique SPE for plurality voting, and a different unique SPE for approval voting. \emph{If the voting method is plurality}, and voter \#1 votes for \#5, then \#5 will be elected (even if voter \#2 decides to vote for \#1, voter \#3 cannot get elected himself and he will vote for \#5 which is his second best outcome). However, voter \#1 has a better vote. He may abstain, thus forcing voter \#2 to vote for him (if voter \#2 votes for \#5 then voter \#3 may abstain and get elected with the help of the vote from \#4 and the tie-breaking). Since \#3 and \#5 will not get more than one vote, voter \#1 now wins by tie-breaking. Thus by abstaining \#1 gets elected, even though \#5 is the most popular member having the most confirmations. Now, \emph{if approval is the voting method} and \#1 abstains, member \#2 may vote just for \#1 and get him elected as before. He may also vote for both \#1 and \#5. In this case \#3 still cannot get elected (due to the tie-breaking order) and will vote for his second-best outcome, namely \#5. We see that if voter \#1 abstains, voter \#2 has two `best ballots' (to vote for \#1 or to both \#1 and \#5) that lead to an outcome he confirms (\#1 or \#5). Later, when we formalize the model (Section \ref{sec: model}), we add a `truth-bias' assumption which states that each voter prefers the vote which is closest to his true confirmation set. Under this assumption member \#2 favours the vote \{\#1, \#5\} over just \{\#1\}. In this case, \#1 does not gain from abstaining; thus, using the `truth-bias' assumption once more, we get that \#1's best vote is to be truthful (i.e. vote for \#5). Everyone else will be truthful as well, and \#5 will be elected.}
\end{exm}

\begin{exm}\label{exm: second}
	\emph{Figure~\ref{fig: cabinet2} shows a network with four voters and at most one confirmation (outgoing edge) for each voter. The voting order and the tie-breaking order are both lexicographic. Since voters \#1 and \#2 cannot get elected (as no one confirms them), if they have a vote which gets the one they confirm elected (\#4 for voter \#1  and \#3 for voter \#2), then this will be their best vote. Otherwise, they will just be truthful (again we assume they are truth-bias) and vote for the one they confirm. \emph{In plurality voting} voter \#1 has no vote which gets \#4 elected, hence both \#1 and \#2 are truthful and \#3 is elected after abstaining. However, \emph{in approval voting}, \#1 can vote for both \#4 and \#2. Since \#2 precedes \#3 in the tie-breaking order, \#3 cannot be elected and will now vote for \#4. Thus, in approval voting, \#4 is elected. }
	\begin{figure}[h]
		\centering
		\includegraphics[width={0.3\textwidth}]{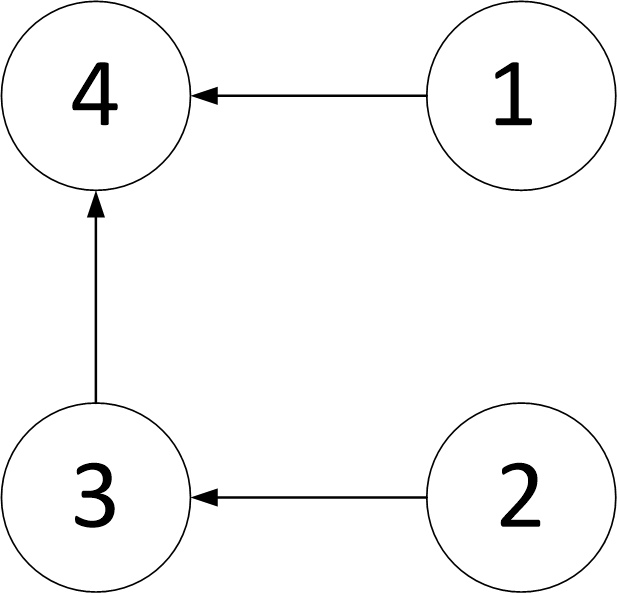}
		\caption{Four committee members and their confirmations.}
		\label{fig: cabinet2}	
	\end{figure}
\end{exm}

\subsection{Related work}
Voting systems and their limitations have been long studied as part of the broader field of social choice (see for example the classic book of Farquharson \cite{farquharson1969theory} and the more recent handbook, \cite{CSC16}). The classical voting model assumes that the sets of voters and alternatives are disjoint, and that each voter has a totally ordered preference over the alternatives. {Sequential voting} with this model has been studied before and showed to contain counter-intuitive `paradoxes'. Desmedt and Elkind \cite{DE10} considered both simultaneous and sequential plurality voting. They showed that a sequential voting system with at least three alternatives is prone to strategic voting, which might lead to an unexpected outcome, such as a Condorcet winner who does not win the election. Conitzer and Xia \cite{CX} further exemplified this phenomenon in a wide range of sequential voting systems, characterized by their domination index. In~\cite{Paradoxes} the authors showed even more extreme examples of strategic ballots in sequential voting systems.

A {confirmation network} as an underlying model for simultaneous voting has also been studied. Holzman and Moulin~\cite{ECTA:ECTA1291} took an axiomatic approach to show the possibilities and limitations of such electoral systems. The main requirement of such systems, in their paper, is that no voter will be able to manipulate the system to select him by delivering a dishonest, strategic ballot. Alon et al.~\cite{AFPT11} investigated the same model, and showed the impossibility of incentive-compatible (that is, `non-manipulable'), deterministic voting systems. They suggested a probabilistic system with a bounded ratio between the maximal in-degree and the expected in-degree of the elected agent. Further works with the same theme can be found in \cite{Bjelde:2017:ISP:3174276.3107922,Fischer:2014:OIS:2600057.2602836,Aziz:2016:SPS:3015812.3015872}.

\subsection{Our contribution}
In this paper, we discuss for the first time sequential voting with the underlying model of a confirmation network, for both plurality and approval voting. We show the limitations of these voting systems by demonstrating extreme cases in which an unpopular member is elected in an SPE. On the other hand, we prove upper bounds on the ratio between the maximal in-degree in the network and the in-degree of an SPE outcome.

\subsection{The model}\label{sec: model}
Let $ A $ be a set of agents. General social choice settings assume that each agent has ordinal preferences over all possible outcomes (i.e., an elected agent in our case). In such general settings it is arguable how to measure the quality of the elected agent. In this paper we restrict attention to a simplified setting where the ordinal preference of each agent has only three levels: each agent prefers himself, those he confirms are second, and those he does not confirm are last. We model the preferences of the agents using a directed graph $ G(A,E) $ where the interpretation of $ (x,y)\in E $ is that agent $ x $ confirms agent $ y $. In this setting we have a natural measure of the quality of the elected agent: the number of incoming edges. As we saw in Examples \ref{exm: first} and \ref{exm: second}, and will see in the results, even in this  simplified setting the strategic analysis is quite involved.

The agents cast their votes sequentially and openly. We consider two voting rules: \emph{plurality}, where each agent is allowed to vote for at most one other agent, and \emph{approval}, in which each agent may vote for any subset of the other agents (abstentions are allowed). The winner of the ballot is the one who receives the most votes, under a predetermined tie-breaking order. We assume $ A $ is ordered; this order is used both as the voting order and the tie-breaking order.\footnote{It is only for convenience that we assume that the voting order and tie-breaking order are the same. The proofs of the upper bounds on the ratio (Theorems~\ref{thm: pluraliry multiplicative bound} and~\ref{thm: approval multiplicative bound}) do not use this assumption. The negative examples can be altered to give the same outcome for different tie-breaking orders (this is not to say that the same example works for all tie-breaking orders).\label{ftn: order}} The utility of agent $ x $ from the outcome $ y $ is
\[ \cl{U}_x(y)=\begin{cases}
1,&y=x\\
1/2,&(x,y)\in E\\
0,& \text{otherwise}.
\end{cases} \]
There is nothing particular about this function; any three-level function will do. Actually, $ \cl{U} $ will not be explicitly used in the remainder of the paper.

We are interested in voting strategies that form a subgame perfect equilibrium (SPE). At least one SPE is guaranteed to exist~\cite{O03}; however, if in some subgame more than one `best vote' option is available to an agent, multiple SPEs exist, possibly with different outcomes. Such a situation can occur, for instance, when an agent does not confirm anyone and is also not confirmed by any other agent (i.e., the agent is an `isolated' node). If many agents are isolated, and so indifferent to the outcome, they may each make an arbitrary vote and anyone may be elected. To avoid this problem, we take the same approach as in~\cite{DUTTA2012154} (see also \cite{Meir2010,thompson2012,OMT13}). Namely, when an agent faces several best-votes, he will select the one which best reflects his true confirmations. In order to quantify the truth-bias assumption, we add the following bonus utility. Let $ f(x) $ be the number of agents that $ x $ confirms and actually votes for, and let $ g(x) $ be the number of agents he does not confirm and nevertheless votes for. Then his bonus utility is 
\[ \cl{B}_x=\epsilon^2 f(x)-\epsilon g(x), \]
where $ 0<\epsilon<1/2n $.\footnote {The intuition behind this function is that a voter would rather not report some or all of his confirmations than vote for someone he does not confirm. The upper bound on $ \epsilon $ is chosen so that the utility from the outcome always dominates the bonus utility.} When the outcome is $ y $, the actual utility of agent $x$ is given by $\cl{U}_x(y) +  \cl{B}_x$.

\subsection{Definitions and notations}

We will use the following notations from graph theory. For $ a\in A $ let $ d(a)=d_{in}(a):=\#\{b\in A:(b,a)\in E \} $ be the \emph{popularity} of $ a $. We denote by $ \Delta_{in}(G):=\max\limits_{a\in A}d(a)$ the maximum in-degree in $ G $. Similarly, $\Delta_{out}(G) $ is the maximum out-degree. An agent $ m $ is \emph{most popular} if $ d(m)=\Delta_{in}(G) $.

An agent is \emph{achievable} if there is an SPE in which he is elected. We denote by $ W\subseteq A $ the set of all achievable agents. Our goal in this paper is to find bounds on the difference and ratio between the highest/lowest popularity of achievable agents and the maximal popularity in the graph. Formally, for an achievable agent $ w\in W $, let $ G_{-E_w} $ be the graph we get from $ G $ after removing all the out-edges of $ w $. Let
\[ \cl{D}(w)=\Delta_{in}(G_{-E_w})-d(w),\quad \cl{R}(w)=\dfrac{\Delta_{in}(G_{-E_w})}{d(w)} \]
be the additive gap and the multiplicative ratio, respectively, between the highest popularity and the popularity of $ w $. We have two justifications for defining these measures on $ G_{-E_w} $ and not directly on $ G $. The first is philosophical: we do not want $ w $'s own confirmations to influence the way he is measured.\footnote{This argument relates to the notion of `incentive compatibility' which is central in \cite{AFPT11} and \cite{ECTA:ECTA1291}.} The second is mathematical elegance. We pay a small price in the definitions in order to get clearer theorems. It is obvious, though, that $ \Delta_{in}(G_{-E_w})\geq \Delta_{in}(G)-1 $; thus it makes little difference, especially for large values of $ \Delta_{in} $.

With a slight abuse of notation, we define for any graph $ G $ with plurality/approval voting, $ \cl{D}(G)=\min\limits_{w\in W}\cl{D}(w) $;\footnote{Our definition for $ \cl{D}(G) $ uses the best winner (and not the worst). This strengthens our results for the additive gap which are all negative.} for either plurality/approval let $ \cl{D}=\sup\limits_{G}\cl{D}(G) $.\footnote{For ease of notation we did not add a subscript to distinguish between $ \cl{D} $ of plurality and $ \cl{D} $ of approval. The results are the same for both anyway.} We shall promptly see that $ \cl{D} $ is unbounded for both plurality and approval voting. In order to give a better description of the limitations of the two voting methods, and to differentiate between them, we would like to chart the \emph{asymptotic} bounds of the multiplicative ratio when $ \cl{D}\to\infty $.\footnote{Note that if we define $ \cl{R} $ without this asymptotic then it will be predominated by small graphs with low values of $ \cl{D} $.}$^,$\footnote{Alternatively, we can take the asymptotic with respect to $\Delta_{in} \rightarrow \infty$.} To that end, we define for a graph $ G $,
\begin{flalign*}
\underline{\cl{R}}(G)=\min\limits_{w\in W}\cl{R}(w),\quad\overline{\cl{R}}(G)=\max\limits_{w\in W}\cl{R}(w),
\end{flalign*}
the maximal/minimal multiplicative ratio between the most popular agent and the winners. For any positive integer $ k $, we denote by $ \cl{G}_k $ the family of graphs with $ \cl{D}(G)\geq k $, and define\footnote{Again we do not have different notations for plurality and approval. It will be clear from the context to which of the two we refer.}$^,$\footnote{$ \underline{\cl{R}} $ and $ \overline{\cl{R}} $ are analogue to the `price of stability' and the `price of anarchy', respectively (see \cite{AGT07} Section 17.1.3).}
\[ \underline{\cl{R}}=\lim\limits_{k\to\infty}\sup\limits_{G\in \cl{G}_k}\underline{\cl{R}}(G),\quad \overline{\cl{R}}=\lim\limits_{k\to\infty}\sup\limits_{G\in \cl{G}_k}\overline{\cl{R}}(G).\]

\subsection{Main results and paper organization}
In Section~\ref{sec: additive gap}, we show a sharp transition of the additive gap. In networks where each agent confirms at most one other agent (i.e., the maximum out-degree is one) there is a unique outcome, and $ \cl{D}(G)$ is always zero (Proposition~\ref{prp: 1-regular}). However, already for networks where agents confirm at most two other agents (i.e., $\Delta_{out}=2$), $ \cl{D} $ is unbounded (Proposition~\ref{prp: max-deg 2}). In Section~\ref{sec: multiplicative ratio}, we prove bounds on $ \underline{\cl{R}}$ and $ \overline{\cl{R}} $. For approval voting we show that $1.5 \leq \underline{\cl{R}} \leq \overline{\cl{R}} \leq 2$ (Theorem~\ref{thm: pluraliry multiplicative bound}); whereas, for plurality voting we show that $\underline{\cl{R}}\leq 2$ and $\overline{\cl{R}}=\infty$ (Theorem~\ref{thm: approval multiplicative bound}). These results indicate that in worst-case scenarios approval voting succeeds in selecting more popular agents than plurality voting. In Section~\ref{sec: k-approval} we sketch a generalization of our results to $ k $-approval voting. We wrap up with a discussion and open problems in Section~\ref{sec: discussion}.

\section{Bounds on the additive gap}\label{sec: additive gap}
We start our discussion with the special case of graphs with a maximum out-degree of one (i.e. each agent confirms at most one other agent). In this case, we show that both plurality and approval voting have a unique winner in any SPE, and the winner is a most-popular agent or almost most-popular agent.\footnote{Meaning, that if we ignore his own confirmations, the elected agent is most popular.}
In approval voting, a vote of an agent to the set of his confirmed agents is called \emph{truthful}. In plurality voting, a vote of an agent to one of his confirmed agents (or abstention if he does not confirm anyone) is called \emph{truthful}.
Our proof is based on the following observation, which is a consequence of our truth-bias assumption.
\begin{obs}\label{obs:}
	In every SPE with outcome $w$, any agent who does not confirm $ w $ is truthful in the SPE path.\footnote{In fact, this simple observation holds even for a wider solution concept of Nash equilibria.}
\end{obs}
The reason is that the election of $ w $ is one of the worst outcomes for any agent who does not confirm $ w $; hence, being truthful is the only best vote for such an agent. 

\begin{prop}\label{prp: 1-regular}
	For the class of graphs with a maximum out-degree of one, both plurality and approval voting have a unique achievable outcome,\footnote{Though the outcome might be different between plurality and approval; see Example~\ref{exm: second}.} and for both plurality and approval $ \cl{D}=0 $.
\end{prop}
Before we prove this proposition, let us exemplify it in the scenario of Example~\ref{exm: second}. Notice that every node in the network of that example has either one out-edge or none, so the proposition applies. Indeed, we showed there that both plurality and approval have a unique SPE. In addition, in approval voting agent \#4 is elected, and he is the most popular. In plurality, agent \#3 is elected; notice, that if we remove his out-edge to \#4, then he becomes one of the most popular agents.
\begin{proof}
	We start by showing that the outcome is unique using backward induction. Given a subgame (i.e., a history of votes), if the agent which is about to vote has a vote which gets him elected, then this will be the outcome. Moreover, if he cannot get elected but he can get the one he confirms elected, then this would be the outcome (here we use the assumption that he confirms at most one agent). If he cannot get elected and cannot get the one he confirms elected, by Observation~\ref{obs:} his unique best action is to be truthful, and by induction, the outcome is determined uniquely. 
	
	Now fix an SPE. Let $ w $ be the winner of this SPE and let $ m\neq w $ be one of the most popular agents. We denote by
	\begin{flalign*}
	C_W:=\{a\in A: (a,w)\in E \};\;
	C_M:=\{a\in A: (a,m)\in E \};
	\end{flalign*}
	the set of agents which confirm $ w $ and those which confirm $ m $, respectively. By our assumption on the out-degree, $ C_M\cap C_W=\emptyset $. Thus, by Observation~\ref{obs:} the agents in $ C_M\backslash \{w\} $ are truthful. So $ m $ gets the votes of all those who confirm him, except perhaps the vote of $ w $. Again by Observation~\ref{obs:}, no agent in $ A\backslash C_W $ votes for $ w $, which means that $ w $ cannot get more votes than his in-degree. Since $ w $ is elected, we reach the conclusion that 
	\[ |C_W|\geq |C_M|- \mathds{1}_{(w,m)\in E}, \]

	and that is exactly the same as $ \cl{D}=0 $.
\end{proof}

The proof of Proposition~\ref{prp: 1-regular} can be generalized to subgames in which the remaining voters confirm at most one agent. Suppose we are in the middle of a voting process with graph $ G $. Let $ U\subset A $ be the voters who have not yet voted and let $ G' $ be the graph we get from $ G $ after removing the out-edges of vertices in $ A\backslash U $. Let $ \overline{s}=(s_1,\ldots,s_n) $ be the current \emph{scoring vector}.\footnote{That is, $ s_i $ is the number of votes agent $ i $ received from the voters in $ A\backslash U $.} We define the \emph{potential} of a vertex $ a\in A $ in this subgame to be $ \rho(a)=d_{in}(a,G')+s(a) $, where $ d_{in}(a,G') $ is the in-degree of $ a $ in $ G' $.\footnote{In other words, $ \rho(a) $ is the maximum number of votes $ a $ can hope to reach when the voting is done.} Let $ \cl{P}=\max\limits_{a\in A}\rho(a) $. 
\begin{prop}\label{cor: 1-regular}
	Using the definitions above, if $ \Delta_{out}(G')\leq 1 $ then there is a unique SPE for the remaining voting process; if $ w $ is the outcome of this SPE and $ m $ is any agent with $ \rho(m)=\cl{P} $, then
	\[ \cl{P}-\rho(w)\leq \mathds{1}_{(w,m)\in E(G')}. \]
\end{prop}
\noindent We omit the proof which is very similar to that of Proposition~\ref{prp: 1-regular}. Proposition~\ref{cor: 1-regular} will be used in the proof of Proposition~\ref{prp: max-deg 2}.\\

In contrast to Proposition~\ref{prp: 1-regular}, we will now show that even for graphs with a maximum out-degree of two, $ \cl{D}$ is no longer bounded. In the proof, we will show a voting scenario in which voters who confirm both a very popular agent and a much less popular one are forced to vote only for the less popular.
\begin{prop}\label{prp: max-deg 2}
	For the class of graphs with a maximum out-degree of two, $ {\cl{D}} $ is unbounded, for both plurality and approval voting.
\end{prop}
\begin{proof}
	\begin{figure}[h]
		\centering
		\includegraphics[width={0.5\textwidth}]{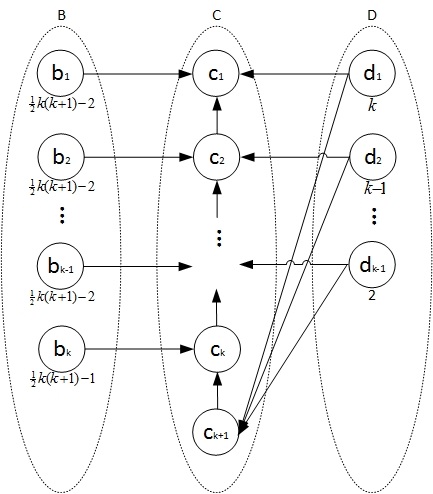}
		\caption{The graph $ G_k $. Agent $ c_1 $ has popularity $ \tfrac{1}{2}k(k+1)+k-1 $ while the winner, $ c_{k+1} $, has popularity $ \tfrac{1}{2}k(k+1)-1 $.}
		\label{fig: increasing gap}	
	\end{figure}
	We will build a series of graphs, $ \{G_k\} $, such that $ \forall k\geq 2 $, $\Delta_{out}(G_k)=2 $, and the unique achievable outcome, for both plurality and approval voting, has popularity $\Delta_{in}(G_k)-k $. Figure~\ref{fig: increasing gap} depicts the graph $ G_k $. The agents in $ B $ and $ D $ are classified by their types (the number of agents in each type is denoted below its circle).
	The order of the agents starts with the agents in $ D $ by lexicographic order of their type, then agents in $ C $ by reverse lexicographic order and finally the agents in $ B $ (i.e., the order is $ d_1,\ldots,d_{k-1},c_{k+1},\ldots,c_1,b_1,\ldots,b_k $).\footnote{As explained in Footnote~\ref{ftn: order}, we use the same order for voting and for tie-breaking.} \\
	Notice that by Observation~\ref{obs:}, the winner in any SPE must be from $ C $; otherwise, we will have a winner which got no votes, while $ c_1 $ gets at least the votes of the agents in $ b_1 $. Suppose we are in the subgame which starts right after the ballot of the last agent in $ D $. Since all the remaining voters have at most one out-edge, according to Proposition~\ref{cor: 1-regular} the winner must be an agent which will have the highest potential if he abstains. Since the tie-breaking order is the same as the voting order, the winner will be \emph{the first} agent which will have the highest potential if he abstains.
	
	Now, the agents in $ D $ confirm both $ c_{k+1} $ and one other agent. The point will be that the only best vote for these agents is to vote only for $ c_{k+1} $. Before proving the general case, we demonstrate this phenomenon in the simplest case, when $ k=2 $ (Figure~\ref{fig: increasing gap k=2}). Here, if the agents of type $ d_1 $ give $ c_1 $ at least one vote (e.g., if one votes for $ c_3 $ and the other for $ \{c_1,c_3\} $), then $ c_3 $ cannot be elected (since after abstaining his total votes will be at most two, and $ c_1 $'s potential is at least three). Therefore $ c_3 $ is truthful and $ c_2 $ abstains and wins (he will have two votes from $ B $ and one from $ c_3 $; agent $ c_1 $ may also have three votes, but $ c_2 $ precedes him in the tie-breaking order). This result is unfavourable for the agents of $ d_1 $. However, if the agents $ d_1 $ vote only for $ c_3 $, then $ c_3 $ can now abstain; having the same potential as $ c_1 $ and $ c_2 $, agent $ c_3 $ wins by the tie-breaking.
	\begin{figure}[h]
		\centering
		\includegraphics[width={0.5\textwidth}]{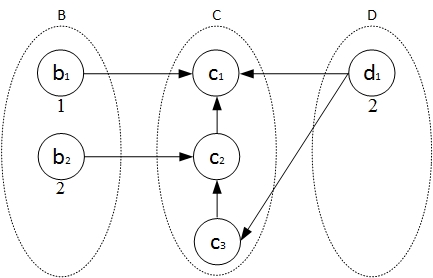}
		\caption{The graph $ G_2 $. }
		\label{fig: increasing gap k=2}	
	\end{figure}	
	
	Turning to the general case, assume first that all the agents in $ d_i $, $ 1\leq i\leq k-1 $ vote for $ c_i $, and perhaps also for $ c_{k+1} $ (in case of approval voting). Since after the ballots of the agents in $ D $, the agents $ c_3,\ldots, c_{k+1} $ all have a lower potential than $ c_1 $, while $ c_2 $ can abstain and have the highest potential, $ c_2 $ is the winner. This outcome is unfavourable for the voters of $ d_1 $. We claim that a better vote for them is to vote only for $ c_{k+1} $, since that leads to the election of $ c_{k+1} $.\footnote{To be more precise: each voter in $ d_1 $ considers the situation in his turn. If all the voters before him voted only for $ c_{k+1} $ then he sees the opportunity to make $ c_{k+1} $ elected. Otherwise, he cannot get a good result and is just truthful.} Indeed, if now all the voters in $ d_i $, $ 2\leq i\leq k-1 $ vote for $ c_i $, then now $ c_1,c_4,\cdots,c_{k+1} $ cannot be elected since $ c_2 $ will have a higher potential than theirs. However, $ c_3 $ can abstain and win by tie-breaking. Thus, the agents of $ d_2 $ are now dissatisfied. If they now all vote just for $ c_{k+1} $ the same reasoning continues and shows that now $ c_4 $ will be the winner unless all the agents of $ d_3 $ vote just for $ c_{k+1} $. Eventually, if all the voters of $ D $ vote just for $ c_{k+1} $ and $ c_{k+1} $ abstains, he will get elected. The agents of $ D $ are all satisfied with this outcome, which shows that this is an equilibrium. Indeed, our reasoning shows that this is the only equilibrium for both plurality and approval voting. The difference between the popularity of the winner, $ c_{k+1} $, and the most popular, $ c_1 $, is $ k $, which implies the claim of the proposition.
\end{proof}

\section{Bounds on the multiplicative ratio}\label{sec: multiplicative ratio}
In Section~\ref{sec: additive gap} we showed that in general, as much as we can tell from the additive gap measure, both plurality and approval voting systems perform poorly. Notice, however, that in the series of graphs in the proof of Proposition~\ref{prp: max-deg 2}, the maximum in-degree in the graph $ G_k $ is $ \Delta_{in}(G_k)=\Theta(k^2) $. This means that the ratio between the maximum in-degree and the popularity of the unique achievable outcome is $ 1+\Theta(k^{-1}) $. This raises the question whether a constant fraction of popularity can be guaranteed in sequential voting. We shall see in Theorems~\ref{thm: pluraliry multiplicative bound}~and~\ref{thm: approval multiplicative bound}, that the bounds of the multiplicative ratio are non-trivial and are quite different between plurality and approval voting.
\begin{thm}\label{thm: pluraliry multiplicative bound}
	In plurality voting, $ \underline{\cl{R}}\leq 2 $ and $ \overline{\cl{R}} $ is unbounded.
\end{thm}
\begin{proof}
	We shall first prove that $ \overline{\cl{R}} $ is unbounded. We show a series of graphs and SPEs, such that the ratio between the most popular agent and the winner goes to infinity. To this end, consider the graph in Figure~\ref{fig: plurality multiplicative bound}.
	Suppose the agents' order is: $ d_1,d_2,d_3, c_3, c_2, c_1, b_1 $. We claim that the following profile of strategies is an SPE.\footnote{Note that we only want to show that it is an SPE; we do not claim uniqueness here.}
	\begin{itemize}[nosep]
		\item Agent $ d_1 $: always vote for $ b_1 $.
		\item Agent $ d_2 $: always vote for $ c_2 $.
		\item Agent $ d_3 $: if $ d_1 $ voted for $ c_1 $, then vote for $ c_2 $. Otherwise, vote for $ c_3 $.
		\item The rest of the agents: be truthful (abstain).
	\end{itemize}
	\begin{figure}[h]
		\centering
		\includegraphics[width=0.4\textwidth]{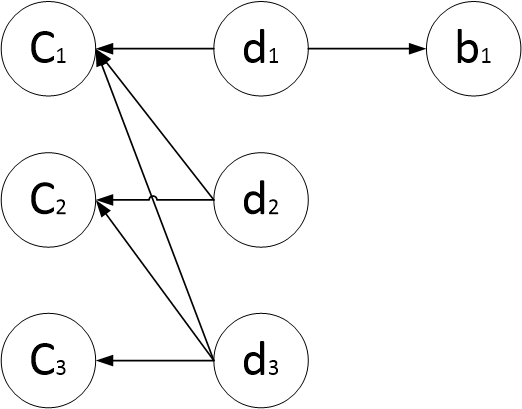}
		\caption{There is an SPE in which $ c_3 $ is elected.}
		\label{fig: plurality multiplicative bound}	
	\end{figure}
	To see that all the agents always act rationally, we start from the last voters and proceed backwards to the first. Agents $ c_1,c_2,c_3,b_1 $ confirm no one, thus, abstaining is always a best vote for them. Agent $ d_3 $ always gets an agent which he confirms elected, so his votes are best possible as well. Moving on to agent $ d_2 $, he will get $ c_2 $ elected when $ d_1 $ votes for $ c_1 $ and that is a best outcome for him. On the other hand, if $ d_1 $ votes for $ b_1 $, then $ d_2 $ knows that $ d_3 $ is about to vote for $ c_3 $, so the result will be bad for him no matter how he votes. The best thing he can do is to vote for someone he confirms (like $ c_2 $). Lastly, agent $ d_1 $ is indifferent between voting for $ b_1 $ and $ c_1 $ because anyhow the elected will be someone he does not confirm ($ c_3 $ in the former case and $ c_2 $ in the latter). Thus, assuming that he votes for $ b_1 $ is legitimate. \\
	This proves the existence of a graph and an SPE with a multiplicative ratio 3. Figure~\ref{fig: plurality multiplicative bound general} shows the general case. Here, there is an SPE in which $ d_1,\ldots,d_{k-1} $ vote for $ b_1,\ldots,b_{k-1} $, respectively; $ d_k $ then votes for $ c_k $, who is elected. If $ d_i $ decides to vote for any of $ c_1,\ldots,c_i $, then $ d_{i+1},\ldots, d_k $ all vote for $ c_{i+1} $, hence $ d_i $ gains nothing.  This is an SPE with a ratio of $ k $, which shows that $ \overline{\cl{R}} $ is unbounded.
	\begin{figure}[h]
		\centering
		\includegraphics[width=0.4\textwidth]{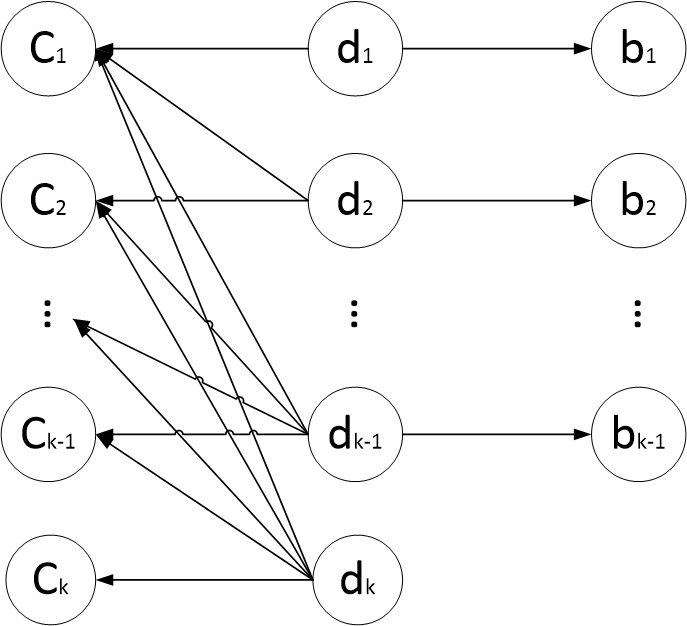}
		\caption{There is an SPE in which $ c_k $ is elected.}
		\label{fig: plurality multiplicative bound general}	
	\end{figure}
	
	In order to prove that $ \underline{\cl{R}}\leq 2 $, we need to show that there is always an SPE in which the winner's in-degree is at least half of $ \Delta_{in} $. Let $ G $ be any graph, and let $ m $ be one of the most popular agents. Assume that every agent who confirms $ m $ would vote for him whenever it is one of his best votes. Fix an SPE with this assumption, and let $ w\neq m $ be the winner. Notice first, that by Observation~\ref{obs:} $ w $ cannot get more than $ d(w) $ votes, since anyone who does not confirm him would not vote for him. Let $ C_{m,w} $ be the set of agents who confirm both $ m $ and $ w $, and let $ C_{m,\overline{w}} $ be the set of agents who confirm $ m $ and do not confirm $ w $. By Observation~\ref{obs:} and our assumption on the SPE, all the agents in $ C_{m,\overline{w}}\backslash\{w\} $ vote for $ m $, which means that $ |C_{m,\overline{w}}\backslash\{w\}|\leq d(w) $. In addition, $ |C_{m,w}|\leq d(w) $. Hence, we get that $ m $'s popularity in $ G_{-E_w} $ is at most $ d(m)= |C_{m,w}|+|C_{m,\overline{w}}\backslash\{w\}|\leq 2d(w) $, and the claim follows.
\end{proof}
In the next theorem, we prove finite bounds for both $ \underline{\cl{R}}$ and $\overline{\cl{R}} $ in approval voting.
\begin{thm}\label{thm: approval multiplicative bound}
	In approval voting, $\dfrac{3}{2}\leq\underline{\cl{R}}\leq\overline{\cl{R}}\leq 2 $.
\end{thm}
\begin{proof}
	The proof of the upper bound on $ \overline{\cl{R}} $ is similar to the proof of the upper bound on $ \underline{\cl{R}} $ in Theorem~\ref{thm: pluraliry multiplicative bound}. This time we do not need to make any assumption on the SPE --- by Observation~\ref{obs:} any voter who confirms the most popular agent and does not confirm the winner is voting for the most popular agent. The claim now follows in a similar way.\\	
	To show the lower bound, we construct a series of graphs $ \{H_k\}_{k\geq 2} $ where $ \Delta_{in}(H_k)=\Theta(k) $ and $ \underline{\cl{R}}(H_k)=3/2 $. In the graph $ H_k $ the agent $ m $ is the most popular, and there are four sets of additional agents:
	\begin{itemize}[wide, labelwidth=!, labelindent=0pt,nosep]	
		\item The $ k $ agents in $ C=\{c_1,\ldots, c_k \} $ are the only agents, besides $ m $, with a positive in-degree. They all have precisely $ k $ in-edges less than $ m $, and all confirm only $ m $. We will show that $ c_1 $ is the winner in the unique SPE.
		\item The $ k-1 $ agents in $ D=\{d_1,\ldots,d_{k-1}\} $ are those who confirm $ m $ but are forced not to vote for him. For any $ 1\leq i\leq k-1 $, agent $ d_i $ confirms $ m $ and all the agents $ \{c_j:j\leq i\} $.
		\item The $ k-1 $ agents in $ E=\{e_2,\ldots,e_{k}\} $ provide the threat which prevents agents of $ D $ from voting for $ m $. Agent $ e_i $, $ 2\leq i\leq k $, confirms all the agents $ \{c_j: j\leq i\} $.
		\item Finally, the set $ B $ contains agents of $ k $ different types which serve as `equalizers' which ensure that all nodes in $ C $ have a popularity of $ d(m)-k=\Delta_{in}(G_k)-k $. For $ 2\leq i\leq k $ there are $ 2i-3 $ agents of type $ b_i $ and they only confirm $ c_i $. In addition, there are $ k-1 $ agents of type $ b_m $ who confirm $ m $.
	\end{itemize}
	\begin{figure}[h]
		\centering
		\includegraphics[width=0.5\textwidth]{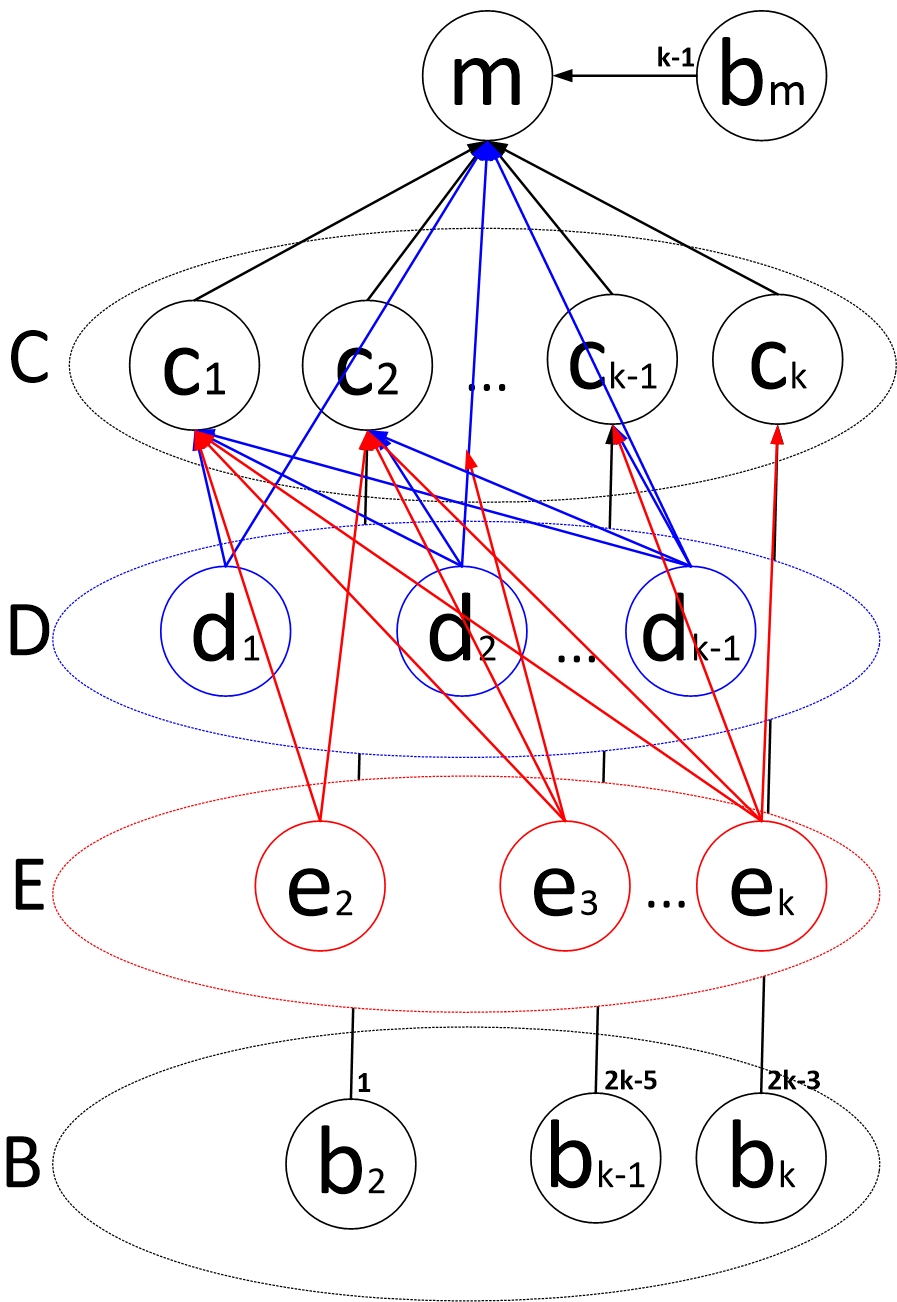}
		\caption{The graph $ H_k $. Agent $ m $ has a popularity lead of $ k $ over all other agents, yet agent $ c_1 $ is elected.}\label{fig: approval k-lead}	
	\end{figure}	
	The general graph is represented in Figure~\ref{fig: approval k-lead}. The agents in $ C,D $ and $ E $ are ordered alternately in lexicographic  order:
	$ c_1,d_1, e_2,c_2,d_2, e_3,\ldots,c_{k-1},d_{k-1}, e_{k},c_k $; the agents in $ B\cup \{m\} $ are ordered after them in arbitrary order. \\
	We will prove by induction on $ k $ that the winner in the unique SPE is agent $ c_1 $. Notice that the popularity of $ m $ in $ G_{-E_{c_1}} $ is $ 3(k-1) $ while the popularity of $ c_1 $ is $ 2(k-1) $, which implies the claimed ratio. Our induction base is $ k=2 $ (Figure~\ref{fig: approval 2-lead}). Here is a sketch of the unique SPE in this scenario. The voting starts with $ c_1 $ who abstains. If $ d_1 $ and $ e_2 $ are both truthful, then $ c_2 $ will be truthful as well (since $ c_1 $ beats him anyhow), and $ m $ will be elected. As $ e_2 $ does not confirm this result, he will vote only for $ c_2 $, who can now abstain and get elected. Agent $ d_1 $, foreseeing this possibility, must vote only for $ c_1 $. Everyone after $ d_1 $ will now be truthful and $ c_1,c_2 $ and $ m $ all end up with two votes, leading to the election of $ c_1 $ by tie-breaking. 
	\begin{figure}[h]
		\centering
		\includegraphics[width=0.4\textwidth]{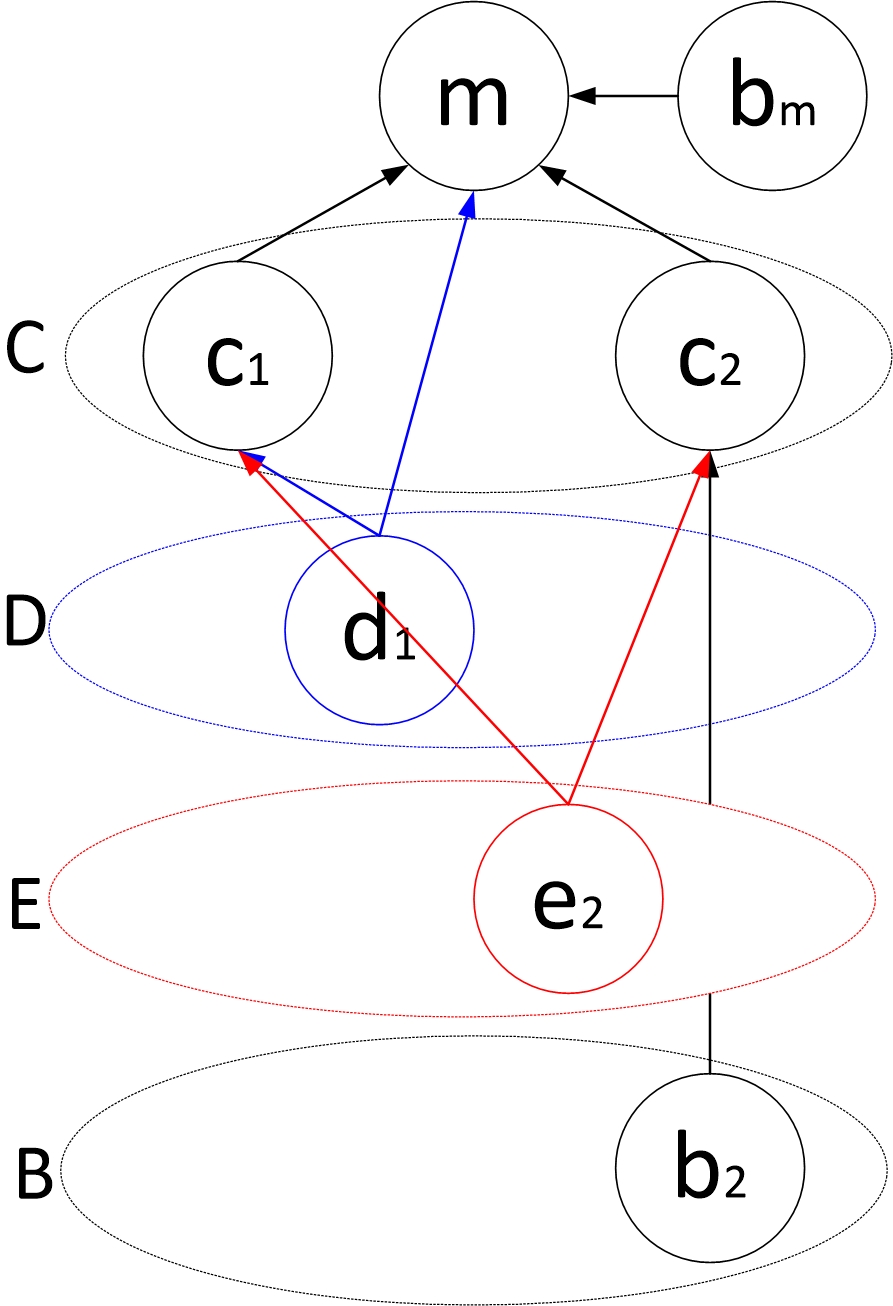}
		\caption{The graph $ H_2 $. Agent $ c_1 $ gets elected after abstaining. }\label{fig: approval 2-lead}	
	\end{figure}				
	For a general $ k $ it is enough to show that the following holds after $ c_1 $ abstains:
	\begin{enumerate}[label=(\arabic*.)]
		\item If $ d_1, e_2 $ are both truthful, then the winner is $ m $.
		\item If $ d_1 $ is truthful and $ e_2 $ votes for $ c_2 $, then $ c_2 $ is the winner.
		\item If $ d_1 $ votes for $ c_1 $ and $ e_2 $ is truthful, then $ c_1 $ wins.
	\end{enumerate}
	Notice that in any case, agents in $ B\cup D\cup E $ are not achievable since they have no in-edges.
	\begin{enumerate}[label=\textit{Proof of (\arabic*.)},wide]
		\item After truthful votes from $ d_1 $ and $ e_2 $, $ c_1 $ and $ c_2 $ have equal potentials. Moreover, every agent who confirms $ c_2 $, except for $ b_2 $, also confirms $ c_1 $. Since in the tie-breaking order $ c_1 $ and $ c_2 $ are adjacent, it is not hard to see that due to the truthful bias assumption, in any SPE, any agent except $ b_2 $ must either include both $ c_1,c_2 $ or neither. This means that in any SPE, $ c_1 $ will precede $ c_2 $ by tie-braking, hence $ c_2 $ is not achievable. It is enough, therefore, to show that if $ c_2 $ and $ d_2 $ are truthful, then $ m $ is the only achievable agent. Indeed, suppose that there is an SPE in which the outcome is $ c_i $, $ 3\leq i\leq k $.\footnote{Clearly, an agent from $ D\cup E\cup B $ cannot get elected, since $ m $ already has three votes.} By Observation~\ref{obs:}, all the agents in $ (C\backslash\{c_1,c_2,c_i\})\cup b_m  $ will be truthful, which means that $ m $ will end up with at least $ 2k-1 $ votes,\footnote{Which are the votes of $ d_1,d_2 $ and $ C\cup b_m\backslash\{c_1,c_i\} $.} while the potential of $ c_i $ is $ 2k-2 $, which is a contradiction.
		\item Notice that by an argument similar to that of (1), here $ c_1 $ is not achievable because $ c_2 $ has higher potential.	We claim that the subgame which starts with the vote of $ c_2 $ is equivalent to the voting (which starts with $ c_1 $) in the graph $ H_{k-1} $. More precisely, we define a mapping $ \varphi:A(H_{k})\backslash (B\cup\{c_1,d_1,e_1\})\to A(H_{k-1})\backslash B $, by $ \varphi(c_j)=c_{j-1} $, $ \varphi(d_j)=d_{j-1} $, $ \varphi(e_j)=e_{j-1} $, $ \varphi(m)=m $. Since the agents of $ B $ vote last and support only one agent each, they are non-strategic.\footnote{That is, their only best vote is to be truthful.} We thus treat these agents as part of the potential of the agents in $ C $ and ignore them in the description of an SPE. Now fix some SPE of the subgame which starts with the voting of $ c_2 $, and denote by $ V(a)\subseteq A(H_k) $ the ballot of agent $ a $ in this SPE. We claim that the profile of voting in $ H_{k-1} $ in which each agent $ \varphi(a) $ vote for $ \varphi(V(a))  $ is an SPE. Indeed, the only difference between the two voting scenarios is that the potential of every agent $ c\in C(H_k)\backslash\{c_1\} $ is higher by exactly two then the potential of $ \varphi(c) $ in $ H_{k-1} $.\footnote{Here the potential includes the ``sure votes'' of agents in $ B $.} Since the graph structure, voting order and tie-braking are all the same, it is not possible that $ \varphi(V(a)) $ is not a best vote of agent $ \varphi(a) $ in $ H_{k-1} $. By the induction hypothesis, we know that the only SPE of $ H_{k-1} $ leads to the election of $ c_1 $. Hence the only achievable agent in our subgame of $ H_{k} $ is $ c_2 $, as claimed.
		\item Just like in (1), $ c_2 $ is not achievable, hence $ c_2 $ is truthful. Again we get a subgame which is equivalent to $ H_{k-1} $, only this time we ignore $ c_2 $ and map $ \varphi(c_1)=c_1 $. The rest of the claims are identical to the proof of (2). We get that the only achievable agent in this subgame is $ c_1 $.
	\end{enumerate}
\end{proof}

\section{Generalizing to {\boldmath $ k $}-approval}\label{sec: k-approval}
The two voting methods we discussed (namely, plurality and approval) can be generalized to a $ k $-approval voting method in which every voter is allowed to vote for \emph{at most} $ k $ other agents.\footnote{We allow a voter to vote for less than $ k $ other agents and even abstain. Though this is not the standard definition of $ k $-approval, since we allowed abstentions in plurality and approval, this is the correct definition to get the full range between plurality and approval with abstentions.} So plurality is no more than 1-approval, and approval is the same as $ (n-1) $-approval. The two propositions of the additive gap (Propositions~\ref{prp: 1-regular} and \ref{prp: max-deg 2}) had a single proof for both plurality and approval, and it is not hard to see that it can be generalized for any $ k $-approval. Likewise, the bound $ \underline{\cl{R}}\leq 2 $ can be proved for any $ k $-approval in a similar way to the proof in Theorem~\ref{thm: pluraliry multiplicative bound} (and see also the same argument in Theorem~\ref{thm: approval multiplicative bound}). We will now extend the second part of Theorem~\ref{thm: pluraliry multiplicative bound} to any $ k $-approval with $ k=o(n) $ by showing a series of graphs in which $ \overline{\cl{R}} $ is unbounded. 

\begin{figure}[h]
	\centering
	\includegraphics[width={0.4\textwidth}]{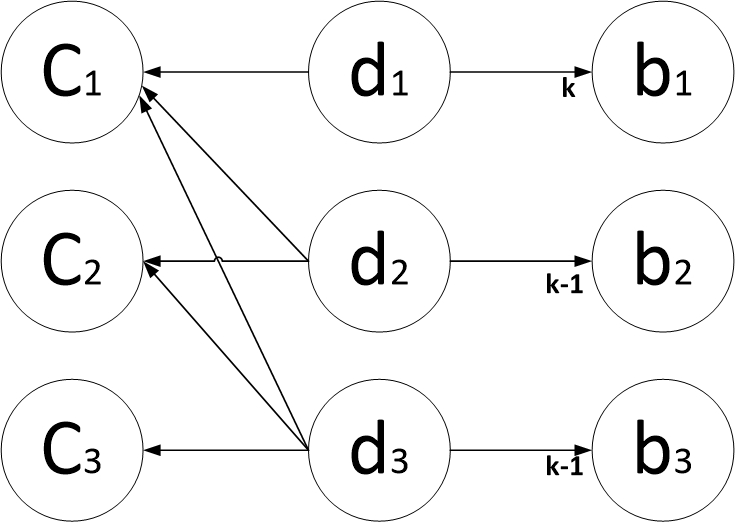}
	\caption{There is an SPE in which $ c_3 $ is elected. }
	\label{fig: k-approval multiplicative ratio 3}	
\end{figure}

In the graph in Figure~\ref{fig: k-approval multiplicative ratio 3} each agent $ d_i $, $ i=1,2,3 $ confirms the agents $ c_j $, $ 1\leq j\leq i $. In addition, agent $ d_1 $ confirms $ k $ additional agents, denoted $ b_1 $, and $ d_2,d_3 $ each confirms $ k-1 $ additional agents, denoted $ b_2,b_3 $, respectively. The agents' order is $ d_1,d_2,d_3, c_3, c_2, c_1 $ and then the rest. We describe an SPE in which $ c_3 $ is elected. In this SPE, if $ d_1 $ votes for any subset which includes $ c_1 $, then $ d_2 $ and $ d_3 $ will both vote for a subset which includes $ c_2 $, and $ c_2 $ is elected. Since the outcome is one of the worst $ d_1 $ can get, voting only for the $ k $ agents of $ b_1 $ is a best vote for him. In this case, $ d_3 $ decides to vote for $ c_3 $ and the $ k-1 $ agents of $ b_3 $, and $ c_3 $ is elected no matter how $ d_2 $ votes. The multiplicative ratio in this SPE is 3. It is not hard to see how to extend this example to get any ratio. 

\section{Discussion and open problems}\label{sec: discussion}

\textbf{Additive gap vs. multiplicative ratio}\\

We have seen (Proposition~\ref{prp: 1-regular}) that in the special case where every agent confirms at most one other agent, the elected agent will be most-popular or almost most-popular. However, as soon as the maximum out-degree of the graph is higher than one, this is no longer the case. In fact, we have shown (Proposition~\ref{prp: max-deg 2} and the discussion in Section~\ref{sec: k-approval}) that $ {\cl{D}} $ is unbounded even when $ \Delta_{in}(G)=2 $, for any $ k $-approval voting. We, therefore, come to the conclusion that the additive gap is not a sufficient measure for the quality of a voting method in this model. Thus we turned to the multiplicative ratio for a finer evaluation. 

Indeed the multiplicative ratio gave us different bounds for plurality and approval voting. In the case of plurality, we have seen (Theorem~\ref{thm: pluraliry multiplicative bound}) that even though there might be an SPE with a `bad' outcome ($ \overline{\cl{R}} $ is unbounded), for every graph, we are guaranteed to have an SPE with a ratio of 2, at most. Moreover, our proof explains how to distinguish this SPE from other SPEs: you just give a small extra incentive for those who confirm the most popular to actually vote for him. The case of approval voting is clearer. Here (Theorem~\ref{thm: approval multiplicative bound}) we have proved finite bounds for both $ \underline{\cl{R}} $ and $ \overline{\cl{R}} $. \\

\noindent \textbf{Plurality vs. approval}\\

In \cite{B08}, Brams demonstrated the superiority of approval voting over plurality voting in simultaneous voting systems. We conclude from our results, that in our setting, plurality voting (and even $ k $-approval voting for any $ k=o(n) $) allows SPEs with unbounded multiplicative ratio, while in approval, this ratio, in any SPE, will be between 1.5 and 2. We cannot draw from our results a comparison of the `best outcome'. To achieve that, we need to bound $ \underline{\cl{R}} $ from below for plurality voting.
\begin{qst}
	In plurality voting, is it possible to construct     a series of graphs, $ \{G_k\} $, with $ \Delta_{in}(G_k)\to\infty $ and $ \underline{\cl{R}}(G_k)\geq \alpha $ for some $ \alpha>1 $ and all $ k $?
\end{qst}
Notice that in the series of graphs of Proposition~\ref{prp: max-deg 2}, when $ {\cl{D}}(G_k)=k $, we have $ \Delta_{in}(G_k)=\Theta(k^2) $. So in this particular example, $ \underline{\cl{R}}(G_k)=\overline{\cl{R}}(G_k)=1+O(1/k) $. If there is a non-trivial (i.e. different than 1) bound for $ \underline{\cl{R}} $ in plurality voting, then there is a series of graphs, $ \{G_k\} $ such that $ {\cl{D}}(G_k)=k $ and $ \Delta_{in}(G_k)=\Theta(k) $ (this is exactly what we have shown for approval voting, when we proved the lower bound in Theorem~\ref{thm: approval multiplicative bound}). So a rephrase of the above question would be:
\begin{qst}
	In plurality voting, is it possible to construct a series of graphs, $ \{G_k\} $, with $ \Delta_{in}(G_k)=\Theta(k) $ and such that $ {\cl{D}}(G_k)\to\infty $?
\end{qst}

It is worthwhile to note here that the example giving the lower bound of Theorem~\ref{thm: approval multiplicative bound} does not work for plurality. To see that, consider the graph $ H_2 $ (Figure~\ref{fig: approval 2-lead}). If $ c_1 $ abstains, then even if $ d_1 $ votes for $ c_1 $, $ e_1 $ might opt to vote for $ c_2 $, and as a result, $ c_2 $ will be elected. So, for the case of plurality voting, the graph $ H_2 $ has an SPE in which $ m $ is elected, and the proof fails.\\

For the approval voting method, we have proved both a lower and upper bound on $ \underline{\cl{R}}, \overline{\cl{R}} $. Still, it could be nice to further narrow these bounds or even find the exact asymptotic values of $ \underline{\cl{R}}, \overline{\cl{R}} $. 

\begin{qst}
	Can the bounds of Theorem~\ref{thm: approval multiplicative bound} be narrowed down?
\end{qst}

\quad\newline
\noindent \textbf{$ \bm{k} $-approval and a threshold between plurality and approval}\\

Finally, we have seen that $ k $-approval has the same bounds as plurality for any $ k=o(n) $. When $ k=n-1 $ this voting method is precisely approval; and so a natural question is what can be said about the threshold function which separates $ k $-approval from plurality.
\begin{qst}
	Find a minimal function, $ f(n) $, such that the voting method $ f(n) $-approval has a finite bound for $ \overline{\cl{R}} $.
\end{qst}

\nocite{*}
\bibliographystyle{eptcs}
\bibliography{paper}

\begin{thebibliography}{10}
\providecommand{\bibitemdeclare}[2]{}
\providecommand{\surnamestart}{}
\providecommand{\surnameend}{}
\providecommand{\urlprefix}{Available at }
\providecommand{\url}[1]{\texttt{#1}}
\providecommand{\href}[2]{\texttt{#2}}
\providecommand{\urlalt}[2]{\href{#1}{#2}}
\providecommand{\doi}[1]{doi:\urlalt{http://dx.doi.org/#1}{#1}}
\providecommand{\bibinfo}[2]{#2}

\bibitemdeclare{inproceedings}{AFPT11}
\bibitem{AFPT11}
\bibinfo{author}{Noga \surnamestart Alon\surnameend}, \bibinfo{author}{Felix
  \surnamestart Fischer\surnameend}, \bibinfo{author}{Ariel \surnamestart
  Procaccia\surnameend} \& \bibinfo{author}{Moshe \surnamestart
  Tennenholtz\surnameend} (\bibinfo{year}{2011}): \emph{\bibinfo{title}{Sum of
  Us: Strategyproof Selection from the Selectors}}.
\newblock In: {\sl \bibinfo{booktitle}{Proceedings of the 13th Conference on
  Theoretical Aspects of Rationality and Knowledge}}, \bibinfo{series}{TARK
  XIII}, \bibinfo{publisher}{ACM}, \bibinfo{address}{New York, NY, USA}, pp.
  \bibinfo{pages}{101--110}, \doi{10.1145/2000378.2000390}.

\bibitemdeclare{article}{AT08}
\bibitem{AT08}
\bibinfo{author}{Alon \surnamestart Altman\surnameend} \&
  \bibinfo{author}{Moshe \surnamestart Tennenholtz\surnameend}
  (\bibinfo{year}{2008}): \emph{\bibinfo{title}{Axiomatic Foundations for
  Ranking Systems}}.
\newblock {\sl \bibinfo{journal}{J. Artif. Int. Res.}}
  \bibinfo{volume}{31}(\bibinfo{number}{1}), pp. \bibinfo{pages}{473--495},
  \doi{10.1613/jair.2306}.

\bibitemdeclare{inproceedings}{Aziz:2016:SPS:3015812.3015872}
\bibitem{Aziz:2016:SPS:3015812.3015872}
\bibinfo{author}{Haris \surnamestart Aziz\surnameend}, \bibinfo{author}{Omer
  \surnamestart Lev\surnameend}, \bibinfo{author}{Nicholas \surnamestart
  Mattei\surnameend}, \bibinfo{author}{Jeffrey~S. \surnamestart
  Rosenschein\surnameend} \& \bibinfo{author}{Toby \surnamestart
  Walsh\surnameend} (\bibinfo{year}{2016}): \emph{\bibinfo{title}{Strategyproof
  Peer Selection: Mechanisms, Analyses, and Experiments}}.
\newblock In: {\sl \bibinfo{booktitle}{Proceedings of the Thirtieth
  AAAIConference on Artificial Intelligence}}, \bibinfo{series}{AAAI'16},
  \bibinfo{publisher}{AAAI Press}, pp. \bibinfo{pages}{390--396}.
\newblock \urlprefix\url{http://dl.acm.org/citation.cfm?id=3015812.3015872}.

\bibitemdeclare{article}{Paradoxes}
\bibitem{Paradoxes}
\bibinfo{author}{Yakov \surnamestart Babichenko\surnameend},
  \bibinfo{author}{Oren \surnamestart Dean\surnameend} \&
  \bibinfo{author}{Moshe \surnamestart Tennenholtz\surnameend}
  (\bibinfo{year}{2018}): \emph{\bibinfo{title}{Paradoxes in Sequential
  Voting}}.
\newblock \urlprefix\url{http://arxiv.org/abs/1807.03979}.

\bibitemdeclare{article}{B05}
\bibitem{B05}
\bibinfo{author}{Marco \surnamestart Battaglini\surnameend}
  (\bibinfo{year}{2005}): \emph{\bibinfo{title}{Sequential voting with
  abstention}}.
\newblock {\sl \bibinfo{journal}{Games and Economic Behavior}}
  \bibinfo{volume}{51}(\bibinfo{number}{2}), pp. \bibinfo{pages}{445 -- 463},
  \doi{10.1016/j.geb.2004.06.007}.

\bibitemdeclare{article}{Bjelde:2017:ISP:3174276.3107922}
\bibitem{Bjelde:2017:ISP:3174276.3107922}
\bibinfo{author}{Antje \surnamestart Bjelde\surnameend}, \bibinfo{author}{Felix
  \surnamestart Fischer\surnameend} \& \bibinfo{author}{Max \surnamestart
  Klimm\surnameend} (\bibinfo{year}{2017}): \emph{\bibinfo{title}{Impartial
  Selection and the Power of Up to Two Choices}}.
\newblock {\sl \bibinfo{journal}{ACM Trans. Econ. Comput.}}
  \bibinfo{volume}{5}(\bibinfo{number}{4}), pp. \bibinfo{pages}{21:1--21:20},
  \doi{10.1145/3107922}.

\bibitemdeclare{book}{B08}
\bibitem{B08}
\bibinfo{author}{Steven~J. \surnamestart Brams\surnameend}
  (\bibinfo{year}{2008}): \emph{\bibinfo{title}{Mathematics and Democracy:
  Designing Better Voting and Fair-Division Procedures}}.
\newblock \bibinfo{publisher}{Princeton University Press},
  \doi{10.1515/9781400835591}.

\bibitemdeclare{inproceedings}{DE10}
\bibitem{DE10}
\bibinfo{author}{Yvo \surnamestart Desmedt\surnameend} \&
  \bibinfo{author}{Edith \surnamestart Elkind\surnameend}
  (\bibinfo{year}{2010}): \emph{\bibinfo{title}{Equilibria of Plurality Voting
  with Abstentions}}.
\newblock In: {\sl \bibinfo{booktitle}{Proceedings of the 11th ACM Conference
  on Electronic Commerce}}, \bibinfo{series}{EC '10}, \bibinfo{publisher}{ACM},
  \bibinfo{address}{New York, NY, USA}, pp. \bibinfo{pages}{347--356},
  \doi{10.1145/1807342.1807398}.

\bibitemdeclare{article}{DUTTA2012154}
\bibitem{DUTTA2012154}
\bibinfo{author}{Bhaskar \surnamestart Dutta\surnameend} \&
  \bibinfo{author}{Arunava \surnamestart Sen\surnameend}
  (\bibinfo{year}{2012}): \emph{\bibinfo{title}{Nash implementation with
  partially honest individuals}}.
\newblock {\sl \bibinfo{journal}{Games and Economic Behavior}}
  \bibinfo{volume}{74}(\bibinfo{number}{1}), pp. \bibinfo{pages}{154 -- 169},
  \doi{10.1016/j.geb.2011.07.006}.

\bibitemdeclare{book}{farquharson1969theory}
\bibitem{farquharson1969theory}
\bibinfo{author}{Robin \surnamestart Farquharson\surnameend}
  (\bibinfo{year}{1969}): \emph{\bibinfo{title}{Theory of voting}}.
\newblock \bibinfo{publisher}{Blackwell}.
\newblock \urlprefix\url{https://books.google.co.il/books?id=fIMHAQAAMAAJ}.

\bibitemdeclare{book}{FM14}
\bibitem{FM14}
\bibinfo{author}{Dan~S. \surnamestart Felsenthal\surnameend} \&
  \bibinfo{author}{Moshe \surnamestart Machover\surnameend}
  (\bibinfo{year}{2014}): \emph{\bibinfo{title}{Electoral Systems: Paradoxes,
  Assumptions, and Procedures (Studies in Choice and Welfare)}}.
\newblock \bibinfo{publisher}{Springer}, \doi{10.1007/978-3-642-20441-8}.

\bibitemdeclare{inproceedings}{Fischer:2014:OIS:2600057.2602836}
\bibitem{Fischer:2014:OIS:2600057.2602836}
\bibinfo{author}{Felix \surnamestart Fischer\surnameend} \&
  \bibinfo{author}{Max \surnamestart Klimm\surnameend} (\bibinfo{year}{2014}):
  \emph{\bibinfo{title}{Optimal Impartial Selection}}.
\newblock In: {\sl \bibinfo{booktitle}{Proceedings of the Fifteenth ACM
  Conference on Economics and Computation}}, \bibinfo{series}{EC '14},
  \bibinfo{publisher}{ACM}, \bibinfo{address}{New York, NY, USA}, pp.
  \bibinfo{pages}{803--820}, \doi{10.1145/2600057.2602836}.

\bibitemdeclare{article}{ECTA:ECTA1291}
\bibitem{ECTA:ECTA1291}
\bibinfo{author}{Ron \surnamestart Holzman\surnameend} \&
  \bibinfo{author}{Hervé \surnamestart Moulin\surnameend}
  (\bibinfo{year}{2013}): \emph{\bibinfo{title}{Impartial Nominations for a
  Prize}}.
\newblock {\sl \bibinfo{journal}{Econometrica}}
  \bibinfo{volume}{81}(\bibinfo{number}{1}), pp. \bibinfo{pages}{173--196},
  \doi{10.3982/ECTA10523}.

\bibitemdeclare{inproceedings}{Meir2010}
\bibitem{Meir2010}
\bibinfo{author}{Reshef \surnamestart Meir\surnameend}, \bibinfo{author}{Maria
  \surnamestart Polukarov\surnameend}, \bibinfo{author}{Jeffrey~S.
  \surnamestart Rosenschein\surnameend} \& \bibinfo{author}{Nicholas~R.
  \surnamestart Jennings\surnameend} (\bibinfo{year}{2010}):
  \emph{\bibinfo{title}{Convergence to Equilibria in Plurality Voting}}.
\newblock In: {\sl \bibinfo{booktitle}{Proceedings of the Twenty-Fourth AAAI
  Conference on Artificial Intelligence}}, \bibinfo{series}{AAAI'10},
  \bibinfo{publisher}{AAAI Press}, pp. \bibinfo{pages}{823--828}.
\newblock \urlprefix\url{http://dl.acm.org/citation.cfm?id=2898607.2898740}.

\bibitemdeclare{book}{CSC16}
\bibitem{CSC16}
\bibinfo{author}{Hervé \surnamestart Moulin\surnameend}
  (\bibinfo{year}{2016}): \emph{\bibinfo{title}{Handbook of Computational
  Social Choice}}.
\newblock \bibinfo{publisher}{Cambridge University Press},
  \doi{10.1017/CBO9781107446984}.

\bibitemdeclare{book}{AGT07}
\bibitem{AGT07}
\bibinfo{author}{Noam \surnamestart Nisan\surnameend}, \bibinfo{author}{Tim
  \surnamestart Roughgarden\surnameend}, \bibinfo{author}{\surnamestart \'{E}va
  Tardos\surnameend} \& \bibinfo{author}{Vijay~V. \surnamestart
  Vazirani\surnameend} (\bibinfo{year}{2007}):
  \emph{\bibinfo{title}{Algorithmic Game Theory}}.
\newblock \bibinfo{publisher}{Cambridge University Press},
  \doi{10.1017/CBO9780511800481}.

\bibitemdeclare{inproceedings}{OMT13}
\bibitem{OMT13}
\bibinfo{author}{Svetlana \surnamestart Obraztsova\surnameend},
  \bibinfo{author}{Evangelos \surnamestart Markakis\surnameend} \&
  \bibinfo{author}{David R.~M. \surnamestart Thompson\surnameend}
  (\bibinfo{year}{2013}): \emph{\bibinfo{title}{Plurality Voting with
  Truth-Biased Agents}}.
\newblock In \bibinfo{editor}{Berthold \surnamestart V{\"o}cking\surnameend},
  editor: {\sl \bibinfo{booktitle}{Algorithmic Game Theory}},
  \bibinfo{publisher}{Springer Berlin Heidelberg}, \bibinfo{address}{Berlin,
  Heidelberg}, pp. \bibinfo{pages}{26--37}, \doi{10.1007/978-3-642-41392-6_3}.

\bibitemdeclare{book}{O03}
\bibitem{O03}
\bibinfo{author}{Martin~J. \surnamestart Osborne\surnameend}
  (\bibinfo{year}{2003}): \emph{\bibinfo{title}{An Introduction to Game
  Theory}}.
\newblock \bibinfo{publisher}{Oxford University Press}.

\bibitemdeclare{inproceedings}{thompson2012}
\bibitem{thompson2012}
\bibinfo{author}{David~R.M. \surnamestart Thompson\surnameend},
  \bibinfo{author}{Omer \surnamestart Lev\surnameend}, \bibinfo{author}{Kevin
  \surnamestart Leyton-Brown\surnameend} \& \bibinfo{author}{Jeffrey
  \surnamestart Rosenschein\surnameend} (\bibinfo{year}{2013}):
  \emph{\bibinfo{title}{Empirical Analysis of Plurality Election Equilibria}}.
\newblock In: {\sl \bibinfo{booktitle}{Proceedings of the 2013 International
  Conference on Autonomous Agents and Multi-agent Systems}},
  \bibinfo{series}{AAMAS '13}, \bibinfo{publisher}{International Foundation for
  Autonomous Agents and Multiagent Systems}, \bibinfo{address}{Richland, SC},
  pp. \bibinfo{pages}{391--398}.
\newblock \urlprefix\url{http://dl.acm.org/citation.cfm?id=2484920.2484983}.

\bibitemdeclare{inproceedings}{CX}
\bibitem{CX}
\bibinfo{author}{Lirong \surnamestart Xia\surnameend} \&
  \bibinfo{author}{Vincent \surnamestart Conitzer\surnameend}
  (\bibinfo{year}{2010}): \emph{\bibinfo{title}{Stackelberg Voting Games:
  Computational Aspects and Paradoxes}}.
\newblock In: {\sl \bibinfo{booktitle}{Proceedings of the Twenty-Fourth AAAI
  Conference on Artificial Intelligence}}, \bibinfo{series}{AAAI'10},
  \bibinfo{publisher}{AAAI Press}, pp. \bibinfo{pages}{921--926}.
\newblock \urlprefix\url{http://dl.acm.org/citation.cfm?id=2898607.2898755}.

\bibitemdeclare{inproceedings}{LXY07}
\bibitem{LXY07}
\bibinfo{author}{Lirong \surnamestart Xia\surnameend},
  \bibinfo{author}{J{\'e}r\^{o}me \surnamestart Lang\surnameend} \&
  \bibinfo{author}{Mingsheng \surnamestart Ying\surnameend}
  (\bibinfo{year}{2007}): \emph{\bibinfo{title}{Sequential Voting Rules and
  Multiple Elections Paradoxes}}.
\newblock In: {\sl \bibinfo{booktitle}{Proceedings of the 11th Conference on
  Theoretical Aspects of Rationality and Knowledge}}, \bibinfo{series}{TARK
  '07}, \bibinfo{publisher}{ACM}, \bibinfo{address}{New York, NY, USA}, pp.
  \bibinfo{pages}{279--288}, \doi{10.1145/1324249.1324286}.

\end{thebibliography}
\end{document}